\documentclass[twoside,web,9pt]{ieeecolor}
\usepackage{generic}
\usepackage{cite}
\usepackage{amsmath,amssymb,amsfonts}
\usepackage{algorithmic}
\usepackage{graphicx}
\usepackage{bm}
\usepackage{mathrsfs}
\usepackage{comment}
\newtheorem{theorem}{Theorem}
\newtheorem{rem}{Remark}[section]

\newtheorem{defin}{Definition}[section]
\newtheorem{ass}{Assumption}
\newtheorem{lem}{Lemma}
\newtheorem{prop}{Proposition}
\newtheorem{corol}{Corollary}
\newenvironment{assumption}{\begin{ass}}{\hfill $\bullet$ \end{ass}}
\newenvironment{remark}{\begin{rem}}{\hfill $\bullet$\end{rem}}
\newenvironment{definition}{\begin{defin}}{\hfill $\bullet$ \end{defin}}

\usepackage{algorithm,algorithmic}
\usepackage{hyperref}

\usepackage{textcomp}
\def\BibTeX{{\rm B\kern-.05em{\sc i\kern-.025em b}\kern-.08em
    T\kern-.1667em\lower.7ex\hbox{E}\kern-.125emX}}
\begin{document}
\title{Backstepping Control of Multidimensional Coupled First-Order Hyperbolic PDEs with Collinear Velocities}
\author{Mohamed Camil Belhadjoudja
\thanks{Mohamed Camil Belhadjoudja is with the Department of Applied
Mathematics, University of Waterloo, 200 University Avenue West, Waterloo, ON,
Canada, N2L 3G1 (e-mail: m2camilb@uwaterloo.ca).}
}

\maketitle

\begin{abstract}
This paper addresses the backstepping boundary stabilization of coupled multidimensional first-order hyperbolic systems. We consider systems whose transport velocity fields are collinear, meaning that each velocity field is a scalar multiple of a common base velocity field. Building upon a recent framework developed for scalar multidimensional first-order hyperbolic equations, we introduce a change of variables, based on characteristic curves defined entirely in the spatial domain, that converts the original multidimensional system into a continuum of coupled one-dimensional first-order hyperbolic systems. By designing a backstepping controller for each system in the continuum representation, and assuming that the transit times of the characteristic curves are uniformly bounded, we achieve finite-time stabilization of the multidimensional system.
\end{abstract}

\begin{IEEEkeywords}
Backstepping, boundary control, multidimensional coupled hyperbolic systems, method of characteristics, finite-time stability.
\end{IEEEkeywords}


\newcommand{\ud}{{\boldsymbol u}}      
\newcommand{\vd}{{\boldsymbol v}}      
\newcommand{\al}{{\boldsymbol \alpha}} 
\newcommand{\be}{{\boldsymbol \beta}}  

\newcommand{\Lp}{\Lambda^{+}}          
\newcommand{\Lm}{\Lambda^{-}}          

\newcommand{\Spp}{\Sigma^{++}}
\newcommand{\Spm}{\Sigma^{+-}}
\newcommand{\Smp}{\Sigma^{-+}}
\newcommand{\Smm}{\Sigma^{--}}

\newcommand{\Om}{\Omega}               
\newcommand{\Rn}{\mathbb{R}^{N}}       
\newcommand{\Gp}{\Gamma^{+}}           
\newcommand{\Gm}{\Gamma^{-}}           
\newcommand{\Go}{\Gamma^{0}}           
\newcommand{\Psii}{\Psi^{-1}}          

\newcommand{\ut}{\tilde{\ud}}          
\newcommand{\vt}{\tilde{\vd}}          
\newcommand{\ub}{\bar{\ud}}            
\newcommand{\vb}{\bar{\vd}}            
\newcommand{\Rst}{R_{\star}}           

\section{Introduction}

First-order hyperbolic equations can be used to model population dynamics, heat and mass transfer, fluid flow, traffic networks, and reactive processes. In many applications, the dynamics involve several interacting variables and are therefore described by systems of coupled hyperbolic equations rather than by a single scalar equation. Furthermore, although one-dimensional models are often used as first approximations obtained by neglecting transverse dynamics, many processes are intrinsically multidimensional.

Over the past decades, major and far-reaching advances have been achieved in the boundary control of one-dimensional first-order hyperbolic equations and systems thereof. Notable developments include, but are not limited to, the introduction of backstepping and Fredholm techniques \cite{krstic_H,traffic_book,coron_paper,first_c,hu,auriol,auriol3,irscheid,ghoussein,ghoussein2,guan,aamo,ensemble0,ensemble1,ensemble2,volt_H,back_rob,krstic_review,coron_system}, and Lyapunov-based methods \cite{lyapu_3,lyapu_4,lyapu_5,lyapu_7,lyapu_PI}. At the same time, fundamental concepts of controllability for hyperbolic systems have been extensively and systematically investigated \cite{coron_control,chitour_1,chitour_2}. In recent years, there have been important extensions of some of the aforementioned results to the multidimensional setting \cite{lyapu_8,lyapu_9,lyapu_10,lyapu_11,liu,meurer,vaz_key2,vaz_2,vaz_3,vaz_4,vaz_5,vaz_6,vaz_7,vaz_8,camil}. In this context, the present paper focuses on the extension of one-dimensional backstepping methods to higher dimensions.

In the recent work \cite{camil_1}, we extended the backstepping method developed for scalar one-dimensional first-order hyperbolic equations in \cite{krstic_H} to scalar first-order hyperbolic equations posed on spatial domains of arbitrary dimension. The central ingredient is a novel change of variables, grounded in the method of characteristics, which transforms the original multidimensional system into a continuum of one-dimensional first-order hyperbolic equations. Our construction departs from the classical formulation of characteristics \cite[Sections 2.1 and 3.2]{evans}. Usually, the characteristic curves associated to a first-order $N$-dimensional (with $N\in \mathbb{N}^*$) evolution hyperbolic equation, with Lipschitz velocity $a:\mathbb{R}^N\to \mathbb{R}^N$, are the curves $\{t\mapsto (\phi(t;x),t)\ : \ x\in \mathbb{R}^N\}$, defined in the space–time domain, where $t\mapsto \phi(t;x)$ is the unique complete solution to $\frac{d\phi}{dt} = a(\phi)$ with initial condition $\phi(0;x)=x$. Along these curves, one can transform the original multidimensional system into a continuum of ordinary differential equations (ODE)s. By contrast, in \cite{camil_1}, we suppress the time variable in the definition of characteristics and instead construct curves entirely within the spatial domain. That is, instead of $\{t\mapsto (\phi(t;x),t)\ : \ x\in \mathbb{R}^N\}$, where $t$ denotes the time variable, we consider the curves $\{s\mapsto \phi(s;x) \ : \ x\in \mathbb{R}^N\}$; see Figure \ref{fig1}. This perspective leads to a family of curves foliating the spatial domain, along which the original multidimensional PDE is reformulated, after introducing an appropriate coordinate system, as a continuum of one-dimensional first-order hyperbolic PDEs, rather than a system of ODEs. This reformulation is particularly well-suited to the application of PDE control techniques, such as backstepping. By stabilizing each equation of the continuum in finite time, and under the assumption that the associated characteristic transit times (the time it takes to flow along the characteristic curve from its entry point to its exit point) are uniformly bounded, we deduce the finite-time stability of the multidimensional system.      

To the best of our knowledge, \cite{camil_1} is the only work to propose a backstepping controller for scalar multidimensional first-order hyperbolic equations with general transport operator. Indeed, other extensions have focused on certain parabolic equations \cite{liu,meurer,vaz_key2,vaz_2,vaz_3,vaz_4,vaz_5,vaz_6,vaz_7,vaz_8,camil}, and on ensembles of hyperbolic equations \cite{ensemble0,ensemble1,ensemble3,ensemble2}, which can be interpreted either as a continuum of one-dimensional equations (as in the terminology adopted in this work) or as multidimensional systems in which the transport operator acts only along a single (Euclidean) spatial coordinate. In the present paper, we extend the approach of \cite{camil_1} to coupled systems of first-order multidimensional hyperbolic equations. More specifically, whereas the transformation introduced in \cite{camil_1} maps a general scalar multidimensional first-order hyperbolic equation into a continuum of independent one-dimensional equations of the type considered in \cite{krstic_H}, we show that a coupled system of multidimensional first-order hyperbolic equations can be transformed, through a similar change of variables, into a continuum of coupled one-dimensional first-order hyperbolic systems of the form studied in \cite{hu,auriol,coron_system}. In the resulting representation, systems corresponding to different elements of the continuum are decoupled, and all couplings are confined to the states within each individual one-dimensional system. The key structural condition enabling this reduction is that the transport velocity fields of all equations in the multidimensional system are collinear. That is, each velocity field can be expressed as a scalar multiple of a common base velocity field. Under this assumption, characteristic curves are defined with respect to the common velocity field, while the scalar multipliers determine the direction and speed of propagation of each state along these characteristics. Once this reduction has been achieved, any of the boundary controllers developed in \cite{hu,auriol,coron_system} can be applied independently to each element of the continuum, yielding finite-time stabilization of the original multidimensional system.

As in \cite{camil_1}, the novelty of our approach does not lie in the control of a continuum of PDEs per se, a topic that has already received considerable attention and led to several important results \cite{ensemble0,ensemble1,ensemble3,ensemble2}. Rather, our contribution is to construct a bridge between the control of multidimensional first-order hyperbolic systems and the control of continua of one-dimensional first-order hyperbolic systems. 

The remainder of this paper is organized as follows. Section \ref{sec:system} introduces the class of multidimensional systems under consideration and presents the standing assumptions. In Section \ref{sec:reduction}, we show how the original multidimensional system can be transformed into a continuum of one-dimensional systems. Section \ref{sec:control} is devoted to the design of the backstepping controller. Finally, the paper concludes with a discussion of future research directions.

\textit{Notations.} Let $N\in \mathbb{N}^*$ and a function $v:\mathbb{R}^N\times [0,+\infty)\to \mathbb{R}$, $(x,t)\mapsto v(x,t)$, with $x=(x_1,x_2,...,x_N)\in \mathbb{R}^N$. We denote the gradient of $v$ by $\nabla v = \left(\frac{\partial v}{\partial x_1},\frac{\partial v}{\partial x_2}, ..., \frac{\partial v}{\partial x_n}\right)^{\top},$ where $\partial v/\partial x_i$ denotes the partial derivative of $v$ with respect to $x_i$. Similarly, we denote by $\partial v/\partial t$ the partial derivative of $v$ with respect to $t$. We denote by $v(\cdot,t)$ the function $x\mapsto v(x,t)$. Given two vectors $x,y\in \mathbb{R}^n$ with $x=(x_1,x_2,...,x_N)^{\top}$ and $y=(y_1,y_2,...,y_N)^{\top}$, we let 
$x\cdot y = x_1 y_1 + x_2 y_2 + ... + x_N y_N$ and $|x| = \sqrt{x_1^2+x_2^2+...+x_N^2}$. Given a matrix $A\in \mathbb{R}^{p\times q}$, with $p,q\in \mathbb{N}$, we denote the element at the $i^{th}$ line and $j^{th}$ column of $A$ either by $A_{ij}$ or $[A]_{ij}$. Finally, given a domain $\Omega \subset \mathbb{R}^N$, we denote by $\mathcal{C}(\Omega;\mathbb{R}^n)$, with $n\in \mathbb{N}^*$, the set of continuous functions $v:\Omega \to \mathbb{R}^n$. Similarly, we denote by $\mathcal{C}^1(\Omega;\mathbb{R}^n)$ the set of continuously differentiable functions $v:\Omega \to \mathbb{R}^n$. Moreover, we denote by $L^2(a,b)$, with $a,b\in \mathbb{R}$ and $a<b$, the set of functions $v:(a,b)\to \mathbb{R}$ such that $|v|_{L^2}=\sqrt{\int_{a}^{b}v(x)^2dx} < +\infty$, and by $\mathcal{C}([0,+\infty);L^2(a,b))$ the set of continuous functions $v:[0,+\infty)\to L^2(a,b)$. 

\begin{figure}
    \centering
    \includegraphics[width=\linewidth]{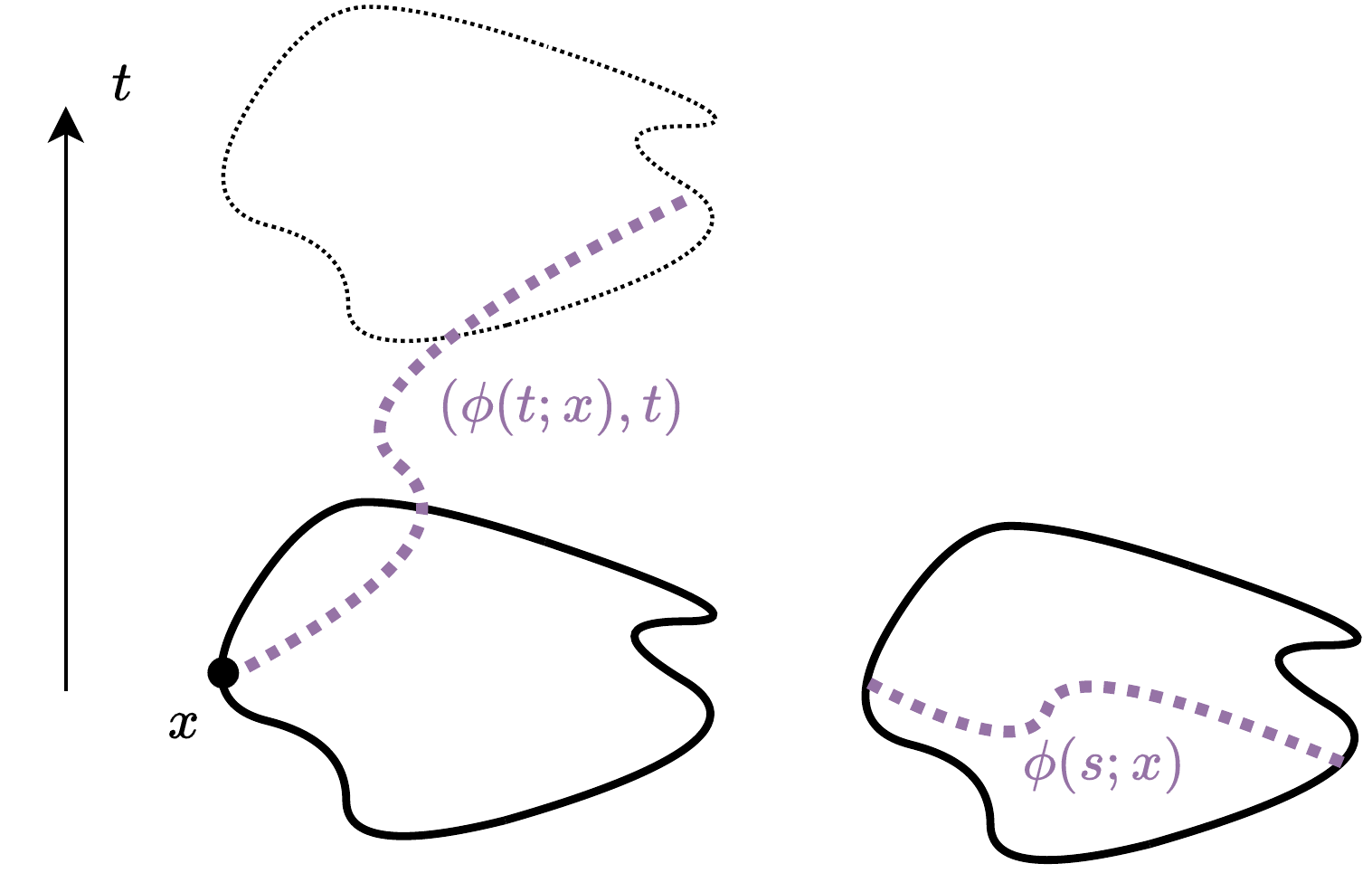}
    \caption{A characteristic curve defined in the space-time domain (left illustration) vs in the space domain (right illustration).}
    \label{fig1}
\end{figure}

\section{System description}\label{sec:system}

This paper extends some of the results in \cite{hu,auriol,coron_system} to arbitrary spatial dimensions.
The systems considered in the aforementioned references are one-dimensional homodirectional/heterodirectional first-order hyperbolic equations of the form
\begin{align}
\frac{\partial \ud}{\partial t} + \Lp\frac{\partial \ud}{\partial x}
  &= \Spp\ud+\Spm\vd, \label{hu1}\\
\frac{\partial \vd}{\partial t} - \Lm\frac{\partial \vd}{\partial x}
  &= \Smp\ud+\Smm\vd, \label{hu2}
\end{align}
where 
\begin{align*}
&\ud(x,t) = (u_1(x,t),\dots,u_n(x,t))^{\top}\in\mathbb{R}^n \\
&\text{and} \ \ \vd(x,t) = (v_1(x,t),\dots,v_m(x,t))^{\top}\in\mathbb{R}^m 
\end{align*}
are the state variables with $(n,m)\in \mathbb{N}\times \mathbb{N}^*$, $x\in [0,1]$ is the
space variable, and $t\geq 0$ is time variable. The diagonal matrices
\begin{align}
&\Lp = \mathrm{diag}(\lambda_1,\lambda_2,\dots,\lambda_n) \in \mathbb{R}^{n\times n}, \label{lambda_+}\\
&\text{and} \ \ \Lm = \mathrm{diag}(\mu_1,\mu_2,\dots,\mu_m) \in \mathbb{R}^{m\times m} \label{lambda_-}
\end{align}
collect the transport velocities, which satisfy
\begin{align}
-\mu_m<\dots<-\mu_1<0<\lambda_1\le\dots\le\lambda_n. \label{speeds}
\end{align}
Furthermore, $\Sigma^{ij}$ are matrices of appropriate dimensions. 

The boundary conditions are
\begin{align}
\ud(0,t)&=Q_0\,\vd(0,t), \\
\vd(1,t)&=R_1\,\ud(1,t)+\bm{U}(\ud(\cdot,t),\vd(\cdot,t)), \label{bc}
\end{align}
where $Q_0,R_1$ are matrices of appropriate dimensions and $\bm{U}(\ud(\cdot,t),\vd(\cdot,t))\in\mathbb{R}^m$ is the control variable, designed to stabilize in finite time the origin $\{(\ud,\vd)=0\}$.

We now formulate an analogous control-design problem in arbitrary dimensions. Let $N\in\mathbb{N}^*$ and let $\Om\subset\mathbb{R}^N$ be bounded, open, and connected.

\begin{assumption}\label{a:dom}
$\Om$ is of class $\mathcal{C}^1$ \cite[p.~626]{evans}.
\end{assumption}

Under Assumption~\ref{a:dom}, $\Omega$ admits a unique unit outward normal
$\nu:\partial\Om\to\mathbb{R}^N$, where $\partial \Omega$ denotes the boundary of $\Omega$.

In one dimension, transport is restricted to two directions along a single axis,
so all velocities are collinear. In higher dimensions one could consider transport
along several axes simultaneously, but this generality leads to difficulties we do
not know how to overcome. We therefore restrict the study to a model with a single distinguished
transport axis. More precisely, let
\begin{align*}
&a : \mathbb{R}^N \to \mathbb{R}^N, \\
&x=(x_1,\dots,x_N)^{\top} \mapsto a(x)=(a_1(x),\dots,a_N(x))^{\top}.
\end{align*}
Transport velocities are assumed to be scalar multiples of $a$:
at each point $x$, transport is either in the direction of $a(x)$ or in the direction of $-a(x)$, depending on the sign of the scalar multiplier.

\begin{assumption}\label{a:lip}
$a$ is Lipschitz.
\end{assumption}

The boundary is decomposed into 
\begin{align*}
\Gp &= \{z\in\partial\Om:\nu(z)\cdot a(z)>0\},\\
\Gm &= \{z\in\partial\Om:\nu(z)\cdot a(z)<0\},\\
\Go &= \{z\in\partial\Om:\nu(z)\cdot a(z)=0\},
\end{align*}

For each $x\in\bar\Om$, where $\bar{\Omega}$ denotes the closure of $\Omega$, the characteristic curve associated to $a$ and passing through $x$ is the complete
classical solution $s \in \mathbb{R}\mapsto\phi(s;x)\in \mathbb{R}^N$ of
\begin{align}
\frac{d\phi}{ds}=a(\phi), \qquad \phi(0;x)=x, \label{char}
\end{align}
whose existence and uniqueness follow from the Cauchy--Lipschitz theorem under
Assumption~\ref{a:lip}. Its \emph{entry time} $\tau^{+}(x)$ and \emph{exit time}
$\tau^{-}(x)$ are defined by
\begin{align*}
\tau^{+}(x) &= \inf\{s\ge 0:\phi(s;x)\in\Gp\}, \\
\tau^{-}(x) &= \sup\{s\le 0:\phi(s;x)\in\Gm\}, 
\end{align*}
whenever they exist.

\begin{assumption}[Non-trapping characteristics]\label{a:nt}
For all $x\in\bar\Om\setminus\Go$, $\tau^{+}(x)$ and $\tau^{-}(x)$ exist and are
finite, and
\begin{align*}
T_{\max}=\sup_{x\in\bar\Om\setminus\Go}\{\tau^{+}(x)-\tau^{-}(x)\}<+\infty.
\end{align*}
\end{assumption}

Since solutions of \eqref{char} are unique, characteristics do not intersect. Hence, each $x\in\bar\Om\setminus\Go$ lies on exactly one characteristic, and the
\emph{transit time} 
$$T(x)=\tau^+(x)-\tau^-(x)>0$$ is constant along it: for
all $\rho\in\Gm$ and all $s\in[0,T(\rho)]$, 
$$T(\phi(s;\rho))=T(\rho).$$

Now, as an extension of \eqref{hu1}--\eqref{bc} to higher dimensions, we consider the system
\begin{align}
\frac{\partial u_i}{\partial t} + \lambda_i\,a(x)\cdot\nabla u_i
  &= \sum_{k=1}^n\Spp_{ik}\,u_k +\sum_{p=1}^m\Spm_{ip}\,v_p, \label{eqU}\\
\frac{\partial v_j}{\partial t} -\mu_j\,a(x)\cdot\nabla v_j
  &=
    \sum_{k=1}^n\Smp_{jk}\,u_k +\sum_{p=1}^m\Smm_{jp}\,v_p, \label{eqV}
\end{align}
for all $i\in \{1,2,\dots,n\}$ and all $j\in \{1,2,\dots,m\}$, with $(n,m)\in \mathbb{N}\times \mathbb{N}^*$, and $u_i(x,t),v_j(x,t)\in \mathbb{R}$.  
\begin{assumption}\label{a:cofol}
The coefficients $\{\lambda_i,\mu_j\}$ verify \eqref{speeds}.     
\end{assumption}

Under Assumption \ref{a:cofol}, $\ud$ is transported in the direction of $a$, whereas $\vd$ is transported in the opposite direction; see Figure \ref{fig2}.

The boundary conditions are prescribed as
\begin{align}
\ud(z,t)&=Q_0(z)\,\vd(z,t), &&z\in \Gm, \label{bcU} \\
\vd(z,t)&=R_1(z)\,\ud(z,t)+\bm{U}(\ud(\cdot,t),\vd(\cdot,t),z), &&z\in \Gp, \label{bcV}
\end{align}
where $\ud=(u_1,u_2,...,u_n)^{\top}$, $\vd=(v_1,v_2,...,v_m)^{\top}$, $Q_{0}:\Gm\to\mathbb{R}^{n\times m}$ and $R_{1}:\Gp\to\mathbb{R}^{m\times n}$, and $\bm{U}(\ud(\cdot,t),\vd(\cdot,t),z)\in \mathbb{R}^{m}$ is the control variable to be designed to stabilize the origin $\{(\ud,\vd)=0\}$ in finite time, in a sense to be made precise below.

From now on, we denote by $\mathbb{S}$ the control system \eqref{eqU}--\eqref{bcV}. Moreover, to simplify the notations, we may write $\bm{U}(z,t)$ instead of $\bm{U}(\ud(\cdot,t),\vd(\cdot,t),z)$. 

\begin{figure}
    \centering
    \includegraphics[width=0.5\linewidth]{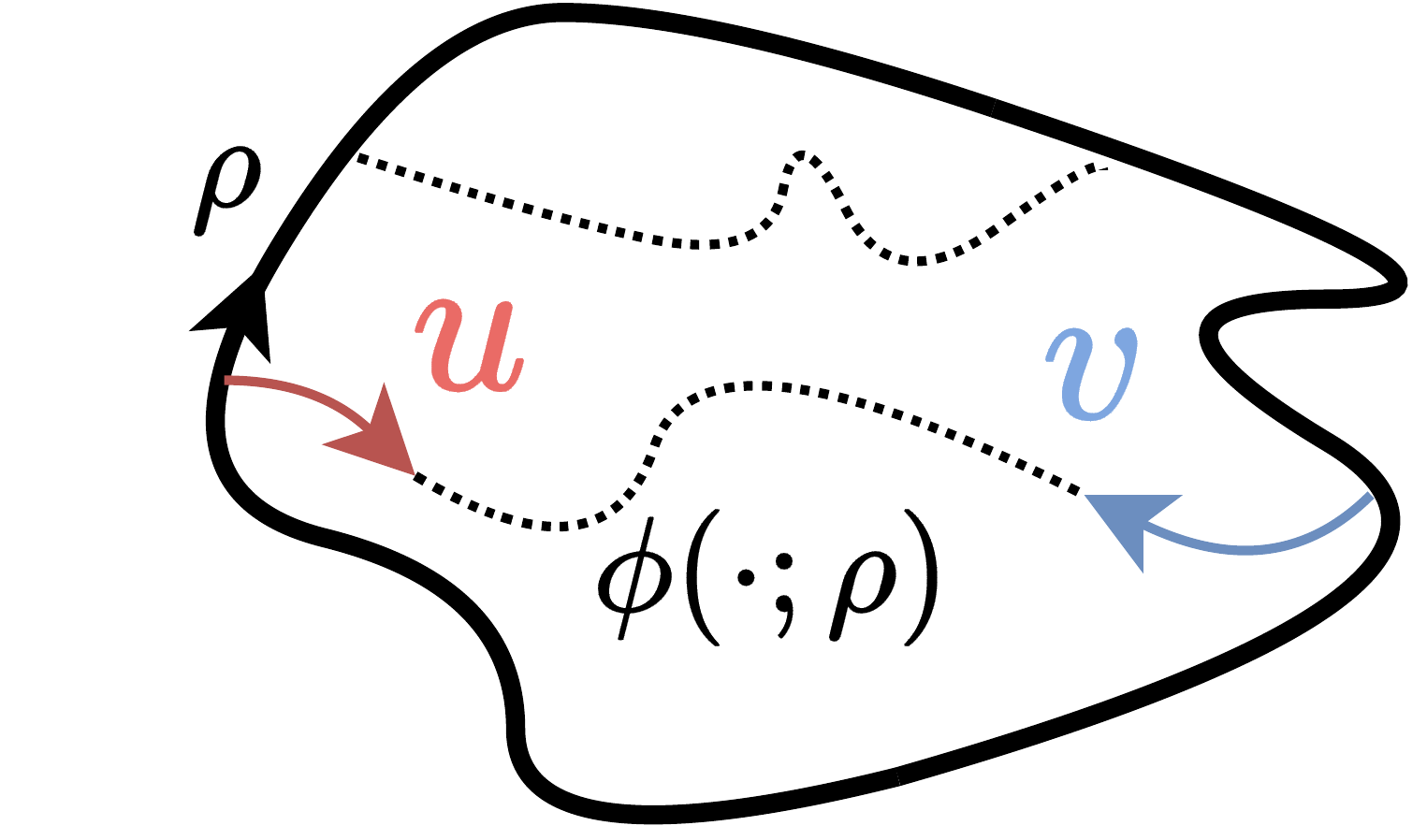}\caption{The states $\ud$ and $\vd$ propagate along the characteristic curves in opposite directions.}
    \label{fig2}
\end{figure}


\section{Concept of solutions and reduction to a continuum of one-dimensional systems}\label{sec:reduction}

Let us first recall the concept of classical solution to $\mathbb{S}$. 
\begin{definition}[Classical solution]\label{def:sol}
A solution to $\mathbb{S}$ with
initial data 
\begin{align*}
(\ud_{o},\vd_{o})\in&~\left[\mathcal{C}(\bar\Om \setminus \Gamma^0;\mathbb{R}^n)\cap\mathcal{C}^1(\Om;\mathbb{R}^n)\right] \\
&~\times \left[\mathcal{C}(\bar\Om\setminus \Gamma^0;\mathbb{R}^m)\cap\mathcal{C}^1(\Om;\mathbb{R}^m)\right]
\end{align*}
is a pair of functions $(\ud,\vd)$ with
$$\ud \in \mathcal{C}((\bar\Om\setminus \Gamma^0)\times[0,\infty);\mathbb{R}^n)\cap\mathcal{C}^1(\Om\times(0,\infty);\mathbb{R}^n),$$
$$\vd \in \mathcal{C}((\bar\Om\setminus \Gamma^0)\times[0,\infty);\mathbb{R}^m)\cap\mathcal{C}^1(\Om\times(0,\infty);\mathbb{R}^m),$$
satisfying \eqref{eqU}--\eqref{eqV} for all $(x,t)\in \Om\times(0,\infty)$, \eqref{bcU} for all $(z,t)\in \Gm\times [0,+\infty)$, \eqref{bcV} for all $(z,t)\in \Gamma^+\times [0,+\infty)$, and $(\ud(x,0),\vd(x,0))=(\ud_o(x),\vd_o(x))$ for all $x\in \bar{\Omega}\setminus \Gamma^0$. 
\end{definition}

Let $(\ud,\vd)$ be a classical solution to $\mathbb{S}$. We will extend the change of variables in \cite{camil_1} to show that the restriction of $(\ud,\vd)$ to each characteristic curve $\phi(\cdot;\rho)$, with $\rho\in \Gamma^-$, solves a one-dimensional system of the same form as \eqref{hu1}--\eqref{bc}. 

For each $x\in\bar\Om\setminus\Go$, we define the \emph{time to exit}
\begin{align*}
\sigma(x) = -\tau^{-}(x) \ge 0, 
\end{align*}
and the \emph{exit point}
\begin{align*}
\rho(x)=\phi(\tau^{-}(x);x)\in\Gm.
\end{align*}
We then introduce the \emph{transformed domain}
\begin{align*}
\hat\Om = \{(\sigma,\rho)\in\mathbb{R}\times\Gm : 0\le\sigma\le T(\rho)\},
\end{align*}
with $T(\rho)$ the transit time of the characteristic $\phi(\cdot;\rho)$, and the map
\begin{align}
\Psi:\bar\Om\setminus\Go\to\hat\Om, \qquad  x\mapsto(\sigma(x),\rho(x)). \label{Psidef}
\end{align}
The map $\Psi$ is a bijection whose inverse is
$\Psii(\sigma,\rho)=\phi(\sigma;\rho)$. We define the
\emph{transformed states}, 
\begin{align}
&\ut(\sigma,\rho,t)=\ud(\Psii(\sigma,\rho),t), \nonumber \\
&\text{and} \quad \vt(\sigma,\rho,t)=\vd(\Psii(\sigma,\rho),t). \label{transdef}
\end{align}
Note, in particular, that 
\begin{align*}
\ut(0,\rho,t)=\ud(\rho,t) \quad \text{and} \quad \vt(0,\rho,t)=\vd(\rho,t).    
\end{align*}
Moreover, 
\begin{align*}
&\ut(T(\rho),\rho,t)=\ud(\phi(T(\rho);\rho),t) \nonumber \\
&\text{and} \quad \vt(T(\rho),\rho,t)=\vd(\phi(T(\rho);\rho),t).
\end{align*}
We also let
\begin{align}
\tilde{R}_1(\rho) := R_1(\phi(T(\rho);\rho)) \in \mathbb{R}^{m\times n}. \label{Rhat}
\end{align}

The following theorem shows that the transformed states satisfy a continuum of one-dimensional systems of the form \eqref{hu1}--\eqref{bc}, indexed by $\rho\in\Gm$. See Figure \ref{fig3} for an illustration of Theorem \ref{thm:reduction}.

\begin{theorem}\label{thm:reduction}
Let Assumptions \ref{a:dom}--\ref{a:nt} hold, and let $(\ut,\vt)$ be defined by \eqref{transdef} with $(\ud,\vd)$ a classical solution to $\mathbb{S}$. Then, for each $\rho\in\Gm$ and all $\sigma\in(0,T(\rho))$,
\begin{align}
\frac{\partial\ut}{\partial t}+\Lp\frac{\partial\ut}{\partial\sigma}
 &= \Spp\ut+\Spm\vt, \label{red:u}\\
\frac{\partial\vt}{\partial t}-\Lm\frac{\partial\vt}{\partial\sigma}
 &= \Smp\ut+\Smm\vt, \label{red:v}
\end{align}
together with the boundary conditions
\begin{align}
\ut(0,\rho,t) &= Q_0(\rho)\,\vt(0,\rho,t), \label{red:bcU}\\
\vt(T(\rho),\rho,t)
 &= \tilde{R}_1(\rho)\,\ut(T(\rho),\rho,t)+\bm U(\phi(T(\rho);\rho),t). \label{red:bcV}
\end{align}
Here, $\Sigma^{ij}$ are the matrices with entries $\Sigma^{ij}_{kp}$ and $\Lambda^+$ and $\Lambda^-$ are defined in \eqref{lambda_+}--\eqref{lambda_-}. 

Conversely, if $(\ut,\vt)(\cdot,\rho,\cdot)$ satisfies
\eqref{red:u}--\eqref{red:bcV} for every $\rho\in\Gm$, and if
$(\ut\circ\Psi,\vt\circ\Psi)$ belongs to
$[\mathcal{C}((\bar\Om\setminus\Go)\times[0,\infty))\cap
\mathcal{C}^1(\Om\times(0,\infty))]^{n+m}$, then
$(\ud,\vd)=(\ut\circ\Psi,\vt\circ\Psi)$ is a classical solution to
$\mathbb{S}$ in the sense of Definition~\ref{def:sol}.
\hfill $\square$
\end{theorem}

\begin{proof}
We transform each term of \eqref{eqU} after the substitution
$x = \Psi^{-1}(\sigma,\rho) = \phi(\sigma;\rho)$; the treatment of \eqref{eqV}
is identical.

\textit{Step 1: time derivative.}
Differentiating both sides of \eqref{transdef} with respect to $t$, and using
the fact that $\Psi^{-1}(\sigma,\rho)$ does not depend on $t$, we obtain
\begin{align}
    \frac{\partial \tilde{u}_i}{\partial t}(\sigma,\rho,t) =
    \frac{\partial u_i}{\partial t}\!\left(\Psi^{-1}(\sigma,\rho),t\right).
\end{align}

\textit{Step 2: transport term.}
We wish to show that
\begin{align}
    a(x) \cdot \nabla u_i(x,t) =
    \frac{\partial \tilde{u}_i}{\partial \sigma}(\sigma,\rho,t)
    \label{p:claim}
\end{align}
at $x = \Psi^{-1}(\sigma,\rho)$. To this end, we differentiate the function
$s \mapsto u_i(\phi(s;x),t)$, to obtain
\begin{align}
    \frac{d}{ds}\, u_i(\phi(s;x),t)
    &= \frac{d\phi(s;x)}{ds} \cdot \nabla u_i(\phi(s;x),t) \nonumber \\
    &= a(\phi(s;x)) \cdot \nabla u_i(\phi(s;x),t).
    \label{p:chain}
\end{align}
Evaluating \eqref{p:chain} at $s = 0$ and using $\phi(0;x) = x$ yields
\begin{align}
    a(x) \cdot \nabla u_i(x,t) =
    \frac{d}{ds}\bigg|_{s=0} u_i(\phi(s;x),t).
    \label{p:identity}
\end{align}
We now set $x = \Psi^{-1}(\sigma,\rho) = \phi(\sigma;\rho)$ in
\eqref{p:identity}. The right-hand side becomes
\begin{align}
    \frac{d}{ds}\bigg|_{s=0}
    u_i\!\left(\phi\!\left(s;\,\phi(\sigma;\rho)\right),t\right).
    \label{p:rhs}
\end{align}
We simplify the argument of $u_i$ using the identity
\begin{align}
    \phi\!\left(s;\,\phi(\sigma;\rho)\right) = \phi(s+\sigma;\rho).
    \label{p:semigroup}
\end{align}
Substituting \eqref{p:semigroup} into \eqref{p:rhs}, we obtain
\begin{align*}
    \frac{d}{ds}\bigg|_{s=0}
    u_i\!\left(\phi\!\left(s;\,\phi(\sigma;\rho)\right),t\right) =
    \frac{d}{ds}\bigg|_{s=0}
    u_i\!\left(\phi(s+\sigma;\rho),t\right).
\end{align*}
By definition \eqref{transdef},
$u_i(\phi(s+\sigma;\rho),t) = \tilde{u}_i(s+\sigma,\rho,t)$.
Substituting and differentiating at $s = 0$ gives
\begin{align}
    \frac{d}{ds}\bigg|_{s=0} \tilde{u}_i(s+\sigma,\rho,t) =
    \frac{\partial \tilde{u}_i}{\partial \sigma}(\sigma,\rho,t).
    \label{p:dsigma}
\end{align}
Combining \eqref{p:identity}--\eqref{p:dsigma}, we conclude \eqref{p:claim}.

\textit{Step 3: coupling and reaction terms.}
As the coefficients
$\{\Sigma^{++},\Sigma^{+-},\Sigma^{-+},\Sigma^{--}\}$ are constant, we have directly
\begin{align*}
    \sum_{k=1}^{n} \Sigma^{++}_{ik}\, u_k\!\left(\Psi^{-1}(\sigma,\rho),t\right)
    = \sum_{k=1}^{n} \Sigma^{++}_{ik}\, \tilde{u}_k(\sigma,\rho,t),
\end{align*}
and likewise for the remaining terms.

\textit{Step 4: boundary condition on $\Gamma^-$.}
At $\sigma = 0$ we have $x = \phi(0;\rho) = \rho \in \Gamma^-$, so that
$\tilde{\ud}(0,\rho,t) = \ud(\rho,t)$ and
$\tilde{\vd}(0,\rho,t) = \vd(\rho,t)$. Substituting in \eqref{bcU} gives
\eqref{red:bcU}.

\textit{Step 5: boundary condition on $\Gamma^+$.}
At $\sigma = T(\rho)$ we have
$x = \phi(T(\rho);\rho) \in \Gamma^+$. Substituting in \eqref{bcV} and using
the definition \eqref{Rhat} gives \eqref{red:bcV}.

\textit{Conclusion.}
Summing the results of Steps~1--3 and using the fact that $(\ud,\vd)$ satisfies
\eqref{eqU}--\eqref{eqV}, we obtain \eqref{red:u}--\eqref{red:v}. The boundary
conditions follow from Steps~4--5.

The converse is immediate by reversing each step and using the assumed
regularity of
$(\tilde{\ud} \circ \Psi, \tilde{\vd} \circ \Psi)$.
\end{proof}

In light of Theorem~\ref{thm:reduction}, we introduce the following concept of
solution.

\begin{definition}[Weak characteristic solution]\label{def:char}
A weak characteristic solution to $\mathbb{S}$ starting from $(\ud_o,\vd_o)$, where
$\sigma\mapsto (\ud_o(\phi(\sigma;\rho)),\vd_o(\phi(\sigma;\rho)))$ belongs to $L^2(0,T(\rho))^{n+m}$ for each
$\rho\in\Gm$, is any pair of real-valued functions $(\ud,\vd)=(\ut,\vt)\circ\Psi$ defined on
$(\bar\Om\setminus\Go)\times[0,\infty)$ such that, for each $\rho\in\Gm$,
$(\sigma,t)\mapsto(\ut,\vt)(\sigma,\rho,t)$ belongs to
$\mathcal{C}([0,\infty);L^2(0,T(\rho))^{n+m})$
and verifies \eqref{red:u}--\eqref{red:bcV} in the weak form with initial condition $\sigma \mapsto (\ud_o(\phi(\sigma;\rho)),\vd_o(\phi(\sigma;\rho)))$.
\end{definition}

\begin{remark}
By the converse in Theorem~\ref{thm:reduction}, any weak characteristic solution that
additionally belongs to
$[\mathcal{C}((\bar\Om\setminus\Go)\times[0,\infty))\cap
\mathcal{C}^1(\Om\times(0,\infty))]^{n+m}$ is a classical solution. Conversely, any classical solution is necessarily a weak characteristic solution, which justifies the proposed concept of solution.
\end{remark}

Finite-time stabilization of $\{(\ud,\vd)=0\}$ reduces to the finite-time stabilization of each one-dimensional system in the continuum representation. It is precisely here that Assumption~\ref{a:nt} is used: it provides a uniform bound on the lengths $T(\rho)$, hence on the settling times across the continuum.

\begin{figure*}
    \centering
    \includegraphics[width=\textwidth]{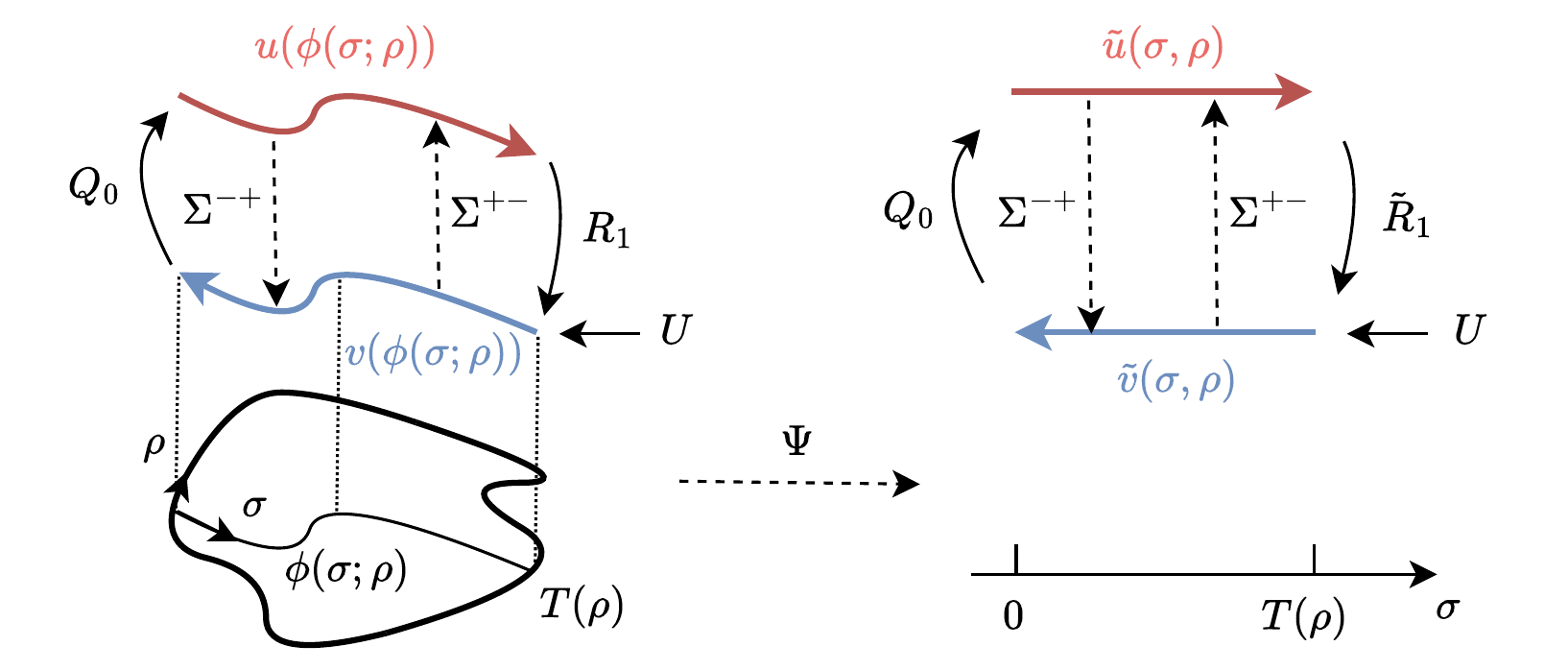}
    \caption{Geometric interpretation of the transformation $\Psi$ defined in \eqref{Psidef}: the characteristic curves are straightened in the transformed coordinates.}
    \label{fig3}
\end{figure*}


\section{Backstepping controller}\label{sec:control}

For each fixed $\rho\in\Gm$ we have a one-dimensional
$(n+m)$ hyperbolic system on $[0,T(\rho)]$ with constant coefficients,
i.e. precisely the class of systems treated in \cite{hu,auriol,coron_system}. We can therefore apply, for
each fixed $\rho$, any of the controllers in \cite{hu,auriol,coron_system}. Let us apply, e.g., the backstepping controller in \cite{auriol}. 

We first rescale $[0,T(\rho)]$ to $[0,1]$. Introduce
\begin{align*}
\bar\sigma=\frac{\sigma}{T(\rho)}\in[0,1], \qquad \bar t=\frac{t}{T(\rho)}\geq 0,
\end{align*}
and the \emph{normalized states}
\begin{align*}
&\ub(\bar\sigma,\rho,\bar t)=\ut(T(\rho)\bar\sigma,\rho,T(\rho)\bar t), \nonumber \\
&\text{and} \quad \vb(\bar\sigma,\rho,\bar t)=\vt(T(\rho)\bar\sigma,\rho,T(\rho)\bar t).
\end{align*}
Then, for each $\rho\in\Gm$, we have
\begin{align}
\frac{\partial\ub}{\partial\bar t}+\Lp\frac{\partial\ub}{\partial\bar\sigma}
 &= T(\rho)\big(\Spp\ub+\Spm\vb\big), \label{nU}\\
\frac{\partial\vb}{\partial\bar t}-\Lm\frac{\partial\vb}{\partial\bar\sigma}
 &= T(\rho)\big(\Smp\ub+\Smm\vb\big), \label{nV}\\
\ub(0,\rho,\bar t) &= Q_0(\rho)\,\vb(0,\rho,\bar t), \label{nbcU}\\
\vb(1,\rho,\bar t) &= \tilde{R}_1(\rho)\,\ub(1,\rho,\bar t)
   +\bar{\bm U}(\rho,\bar t), \label{nbcV}
\end{align}
where 
$$\bar{\bm U}(\rho,\bar t)=\bm U(\phi(T(\rho);\rho),T(\rho)\bar t).$$

According to \cite{hu}, the controller for \eqref{nU}--\eqref{nbcV} is given by 
\begin{align}
\bar{\bm U}(\rho,\bar t)
 =&~ -\tilde{R}_1(\rho)\,\ub(1,\rho,\bar t)
  +\int_0^1\!\big[K(1,\xi;\rho)\,\ub(\xi,\rho,\bar t)
    \nonumber \\
   &~+L(1,\xi;\rho)\,\vb(\xi,\rho,\bar t)\big]\,d\xi, \label{ctrl:norm}
\end{align}
where the backstepping kernels
\begin{align*}
K(\bar{\sigma},\xi;\rho)\in\mathbb{R}^{m\times n},\qquad
L(\bar{\sigma},\xi;\rho)\in\mathbb{R}^{m\times m},
\end{align*}
defined, for each $\rho\in \Gamma^-$, on the triangle 
$$\mathcal{T}=\{(\bar\sigma,\xi):0\le\xi\le\bar\sigma\le1\},$$
satisfy, in the weak sense,
\begin{align*}
\mu_i\frac{\partial K_{ij}}{\partial \bar{\sigma}}&-\lambda_j\frac{\partial K_{ij}}{\partial \xi}
 = T(\rho)\left(\sum_{k=1}^{n}\Spp_{kj}K_{ik}+\sum_{p=1}^{m}\Smp_{pj}L_{ip}\right) \nonumber \\
 &~\quad - \sum_{p=i}^{m}K_{pj}\omega_{ip}(\bar{\sigma};\rho) \quad 1\le i\le m,\ 1\le j\le n, \\
\mu_i\frac{\partial L_{ij}}{\partial \bar{\sigma}}&+\mu_j\frac{\partial L_{ij}}{\partial \xi}
 = T(\rho)\left(\sum_{p=1}^{m}\Smm_{pj}L_{ip}
    +\sum_{k=1}^{n}\Spm_{kj}K_{ik}\right) \nonumber \\
 &~\quad -\sum_{p=i}^m L_{pj}\omega_{ip}(\bar{\sigma};\rho) \quad 1\le i\le m,\ 1\le j\le m, 
\end{align*}
along with the boundary conditions
\begin{align*}
K_{ij}(\bar\sigma,\bar\sigma;\rho)
 &= -\frac{T(\rho)\,\Smp_{ij}}{\mu_i+\lambda_j} \qquad 1\le i\le m,\ 1\le j\le n, \\
L_{ij}(\bar\sigma,\bar\sigma;\rho)
 &= -\frac{T(\rho)\,\Smm_{ij}}{\mu_i-\mu_j} \qquad 1\le i, j\le m, \ j<i, \\
\mu_j\,L_{ij}(\bar\sigma,0;\rho)
 &= \sum_{k=1}^{n}\lambda_k\,[Q_0(\rho)]_{kj}\,K_{ik}(\bar\sigma,0;\rho) \ \ 1\le i, j\le m.
\end{align*}
For each $\rho\in \Gamma^-$, the functions $\omega_{ij}$ with $i\leq j$ are bounded and verify 
\begin{align*}
\omega_{ij}(\bar{\sigma};\rho) = (\mu_i-\mu_j)L_{ij}(\bar{\sigma},\bar{\sigma};\rho)+\Sigma_{ij}^{--} \quad \forall \bar{\sigma} \in [0,1]. 
\end{align*}

Invoking \cite[Theorem 2]{auriol}, we conclude that, for each $\rho\in \Gamma^-$,
\begin{align*}
(\bar{\ud}(\bar{\sigma},\rho,\bar{t}),\bar{\vd}(\bar{\sigma},\rho,\bar{t})) = 0 \quad \forall \bar{\sigma}\in [0,1], \ \forall \bar{t}\geq t_F,
\end{align*}
where 
\begin{align*}
t_F = \frac{1}{\lambda_1}+\frac{1}{\mu_1}. 
\end{align*}
Hence, for each $\rho \in \Gamma^-$, we have 
\begin{align*}
(\tilde{\ud}(\sigma,\rho,t),\tilde{\vd}(\sigma,\rho,t)) = 0 \quad \forall \sigma\in [0,T(\rho)], \ \forall t\geq T(\rho)t_F.
\end{align*}

We now express the control law \eqref{ctrl:norm} in the original variables. Setting $t=T(\rho)\bar t$
and $\sigma=T(\rho)\xi$, so that $d\xi=d\sigma/T(\rho)$ and
$\ub(\xi,\rho,\bar t)=\ut(\sigma,\rho,t)$, $\vb(\xi,\rho,\bar t)=\vt(\sigma,\rho,t)$,
we obtain
\begin{align*}
\bm U(&\phi(T(\rho);\rho),t)
 \nonumber \\
 =&~ -\tilde{R}_1(\rho)\,\ut(T(\rho),\rho,t) +\int_0^{T(\rho)}\!\big[\mathcal{K}^{u}(\sigma;\rho)\,\ut(\sigma,\rho,t)
    \nonumber \\
    &~+\mathcal{K}^{v}(\sigma;\rho)\,\vt(\sigma,\rho,t)\big]d\sigma,
\end{align*}
where  $\mathcal{K}^u(\cdot;\rho)\in\mathbb{R}^{m\times n}$ and
$\mathcal{K}^v(\cdot;\rho)\in\mathbb{R}^{m\times m}$ are given by
\begin{align*}
&\mathcal{K}^{u}(\sigma;\rho)=\frac{1}{T(\rho)}\,K\Big(1,\frac{\sigma}{T(\rho)};\rho\Big), \nonumber \\
&\text{and} \quad \mathcal{K}^{v}(\sigma;\rho)=\frac{1}{T(\rho)}\,L\Big(1,\frac{\sigma}{T(\rho)};\rho\Big).
\end{align*}
Now, for $z\in\Gp$, let the \emph{upstream arc}
\begin{align*}
C(z)=\{\phi(s;z):\tau^{-}(z)\le s\le0\}.
\end{align*}
As $\sigma$ runs over $[0,T(\rho(z))]$, the point $y=\phi(\sigma;\rho(z))$ traces $C(z)$, with
arc-length element $d\ell(y)=|a(y)|\,d\sigma$. Moreover
$\ut(\sigma,\rho(z),t)=\ud(y,t)$ and $\vt(\sigma,\rho(z),t)=\vd(y,t)$. Finally, since
$\phi(T(\rho(z));\rho(z))=z$, one has $\tilde{R}_1(\rho(z))=R_1(z)$. Hence,
\begin{align}
\bm U(z,t)=&~ -R_1(z)\,\ud(z,t)
  +\int_{C(z)}\!\big[\mathcal{G}^{u}(z,y)\,\ud(y,t)
   \nonumber \\
   &~+\mathcal{G}^{v}(z,y)\,\vd(y,t)\big]d\ell(y), \label{ctrl:x}
\end{align}
where
\begin{align*}
\mathcal{G}^{u}(z,y)=\frac{\mathcal{K}^{u}(\sigma(y);\rho(z))}{|a(y)|},
\quad
\mathcal{G}^{v}(z,y)=\frac{\mathcal{K}^{v}(\sigma(y);\rho(z))}{|a(y)|}.
\end{align*}

We have thus proved the following result. 

\begin{theorem}
Consider the system $\mathbb{S}$ and let Assumptions \ref{a:dom}--\ref{a:cofol} hold. Furthermore, let the control input $\bm U$ be given by \eqref{ctrl:x}. Then, for any initial condition $(\ud_o,\vd_o)$, where
$\sigma\mapsto (\ud_o(\phi(\sigma;\rho)),\vd_o(\phi(\sigma;\rho)))$ belongs to $L^2(0,T(\rho))^{n+m}$ for each
$\rho\in\Gm$, there exists a unique weak characteristic solution $(\ud,\vd)$ to $\mathbb{S}$ such that 
\begin{align}
(\ud(x,t),&\vd(x,t)) = 0 \nonumber \\
&~\ \forall x\in \bar{\Omega}\setminus \Gamma^0, \ \forall t\geq T_{max}\left(\frac{1}{\lambda_1}+\frac{1}{\mu_1}\right). 
\end{align}
\hfill $\square$
\end{theorem}

\section{Conclusion}
The collinearity assumption on the velocity fields is admittedly restrictive. In one space dimension, transport necessarily occurs along a single spatial axis, so that all velocity fields are automatically collinear: they can differ only in magnitude and direction. In higher dimensions, however, this is no longer the case: the velocity fields may not be collinear, and the notion of hyperbolicity accommodates far more general systems than those considered here. As such, the present work is to be understood only as a first step toward the backstepping control of general coupled first-order hyperbolic systems in multiple space dimensions. A distinctive feature of our approach, both in the scalar case treated in \cite{camil_1} and in the coupled setting of the present work, is that, once the multidimensional system has been transformed into a continuum of one-dimensional systems via our change of variables, the elements of the continuum are completely decoupled from one another. Consequently, we do not need to exploit the continuum structure any further: each one-dimensional system is stabilized independently, and no interaction between distinct elements of the continuum needs to be accounted for. In particular, we do not need to resort to the techniques developed for the control of ensembles of hyperbolic equations, where coupling occurs across all elements of the ensemble \cite{ensemble0,ensemble1,ensemble3,ensemble2}. However, these works are inspirational, as they show that one can work with ensembles of hyperbolic equations that are coupled, which provides a possible reference point when attempting to extend our framework beyond the collinear case. Indeed, through our change of variables, one can see that the general case with non-collinear velocities reduces to handling certain couplings in the ensemble representation. Whether and how existing ensemble control techniques can be extended to this setting remains an open question.

\end{document}